\documentclass[journal,onecolumn,web]{ieeecolor}

\usepackage{generic}
\usepackage{amsmath,amssymb,amsfonts}
\usepackage{subcaption}
\usepackage{graphicx}
\usepackage{textcomp}

\let\labelindent\relax
\usepackage{enumitem}

\usepackage{algpseudocode}
\usepackage{algorithm}

\usepackage[sorting=none]{biblatex}
\newcommand{\bmat}{\begin{bmatrix}}
\newcommand{\emat}{\end{bmatrix}}
\newcommand{\bsmat}{\left[ \begin{smallmatrix}}
\newcommand{\esmat}{\end{smallmatrix} \right]}

\newcommand{\R}{\mathbb{R}}
\newcommand{\C}{\mathbb{C}}

\newcommand{\T}{\mathrm{T}}

\renewcommand{\Re}{\operatorname{Re}}
\renewcommand{\Im}{\operatorname{Im}}

\newcommand{\RHi}{\mathcal{RH}_\infty}
\newcommand{\RLi}{\mathcal{RL}_\infty}
\renewcommand{\L}{\mathcal{L}}

\newcommand{\ophi}{\overline{\phi}}
\newcommand{\uphi}{\underline{\phi}}

\newtheorem{lemma}{Lemma}
\newtheorem{remark}{Remark}
\newtheorem{theorem}{Theorem}

\newtheorem{problem}{Problem}

\newtheorem{definition}{Definition}

\DeclareMathOperator{\diag}{diag}
\DeclareMathOperator{\blkdiag}{blkdiag}

\begin{document}
\title{Phase Robustness Analysis for Structured Perturbations in MIMO LTI Systems}
\author{Luke Woolcock and Robert Schmid
\thanks{Luke Woolcock and Robert Schmid are with the Department  of Electrical  and  Electronic Engineering, University of Melbourne (e-mail: lpwoolcock@student.unimelb.edu.au; rschmid@unimelb.edu.au). }}

\maketitle

\begin{abstract}
The stability of interconnected linear time-invariant systems using  singular values and the small  gain theorem has been studied for many decades. The  methods of $\mu$-analysis and   synthesis has been  extensively developed to provide robustness guarantees for a   plant subject to structured perturbations,  with  components  in the structured perturbation  satisfying a   bound  on their  largest singular  value.
Recent results on phase-based stability measures have led  to a counterpart of the small gain theorem, known as the small phase theorem. To  date these phase-based methods  have  only  been used to  provide   stability  robustness measures for  unstructured perturbations. 

In this paper, we define a phase robustness metric for multivariable linear time-invariant systems in the presence of a structured perturbation. We demonstrate its relationship to a certain class of multiplier functions for integral quadratic constraints, and show that a upper bound can be calculated via a linear matrix inequality problem. When  combined with  robustness  measures from the  small  gain  theorem,  the new methods are able  provide  less  conservative  robustness  metrics than  can  be  obtained  via   conventional $\mu$-analysis methods. 

\end{abstract}


\section{Introduction}
\label{sec:introduction}

\IEEEPARstart{W}{e} consider the robust stability of the stable multiple-input multiple-output (MIMO) linear time-invariant (LTI) plant $G\in\RHi^{n\times n}$ with respect to a stable LTI perturbation $\Delta\in\RHi^{n\times n}$, as shown in Figure \ref{fig:iqc_feedback}. The signals $u_1,u_2$ represent exogenous inputs to the system and $e_1,e_2$ represent internal signals. The \emph{feedback interconnection} so defined is denoted by $(G,\Delta)$.

The problem at hand is to quantify the robustness of the stability of the perturbed system in terms of a metric defining the ``size'' of the set of $\Delta$ that is guaranteed to result in a stable interconnection. This defines a stability margin for $G$, generalizing the classical notion of gain and phase margin for single-input single-output (SISO) systems.

The most famous result for the  robust stability analysis of feedback interconnections is the small gain theorem \cite{Desoer1975}, which defines the set of admissible perturbations in terms of their $\mathcal{L}_2$-gain. In the case of LTI systems, the $\mathcal{L}_2$ gain is linked to the frequency response by means of the singular values of their transfer matrices, forming the basis of $\mathcal{H}_\infty$ theory. Therefore, for LTI $G$ and $\Delta$, the small gain criterion may be stated as a frequency domain inequality on the singular values of the two subsystems.

\begin{figure}[!h]
	\centering
	\includegraphics{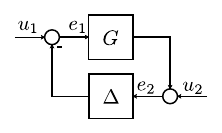}
	\caption{Block diagram of feedback interconnection $(G,\Delta)$.}
	\label{fig:iqc_feedback}
\end{figure}

In practice, much more is known about the perturbation than only its gain, leading to a high level of conservatism when applying the small gain theorem naively. In the case of MIMO systems, we can often assume the perturbation satisfies a certain block diagonal \emph{structure} $\mathbf{\Delta}\subset\RHi^{n\times n}$. Doyle \cite{Doyle1982} addressed the structured robust stability problem by introducing the \emph{structured singular value} $\mu$, a real-valued function on complex matrices that quantifies the minimum destabilizing $\mathcal{L}_2$-gain of the perturbation under the assumption of structure. This gave rise to a broad range of analysis and synthesis techniques with wide applications. In the field of power systems, $\mu$-synthesis techniques have been applied to grid connected inverters with multiparameter distributions \cite{Chen2022}, $\mathcal{H}_\infty$-robust current control for wind turbine generators \cite{Wang2017}  and  the region-wise  small-signal stability of power systems under perturbations \cite{Pan2018}.

When considering the robustness of SISO systems, both the gain and phase margin are considered in conjunction. The small gain theorem straightforwardly provides a notion of gain for MIMO systems. However, an appropriate notion of phase for MIMO systems is much less straightforward to develop. The $\mathcal{L}_2$-gain of a system is invariant under isometries of the input or output -- i.e. it generalizes the property that the gain of a SISO system is invariant when multiplied by a complex number of unit magnitude. It would be desirable to obtain a notion of phase for MIMO systems which has the dual property--that this ``phase'' would be invariant under some type of scaling of the input and output signals by a positive operator. Additionally, in the context of the robustness of the feedback interconnection system depicted in Figure \ref{fig:iqc_feedback}, this notion of phase should lead to a method for the calculation of phase bounds for a set of perturbations that is invariant under such a positive scaling. However, such a notion has proven difficult to formulate. In the following we first review some  historical attempts to formalize phase and phase-based robustness for MIMO LTI systems, and then introduce the phase concepts that will be  employed in our results. 

The well-known passivity theorem \cite{Desoer1975} generalizes, in an imprecise sense,  the notion of phase to MIMO and nonlinear systems. In the LTI case, the notion of positive realness links the properties of a system's transfer function to its passivity. The input and output passivity indices provide a robustness criterion \cite{Bao2007}. The property of passivity is invariant under scaling, as desired, but the passivity \emph{indices} do not generalize this property of SISO phase margins.

An early attempt at defining the phases of a transfer matrix was made by \cite{Postle1981} using the eigenvalues of the unitary part of the polar decomposition. This leads to a frequency-wise stability criterion -- a so-called ``small phase theorem.'' Recently \cite{Pak2020} made use of this definition of phase response to develop  an optimal  controller synthesis under mixed gain and phase constraints. While this approach defines scaling-invariant phases for an arbitrary complex matrix, the resulting small phase theorem depends upon a gain-dependent ``phase modifier'' factor. In the context of robustness analysis, this means that the phase-bounded set of perturbations $\mathbf{\Delta}$ is not invariant under scaling, except by a positive scalar.

Another early contibution to notions of phase for MIMO LTI systems appeared in   \cite{Freudenberg1988},  employing the singular vectors associated with the singular values of the transfer matrix. This  approach provided a  MIMO generalisation  of the Bode gain-phase relationship for SISO systems, a concept that was  further developed in \cite{Chen1998}. However, these works did not present a criterion for the robustness of a plant to a phase-bounded uncertainty, and this definition of phase is not invariant under positive scaling.

\cite{BarOn1990} applied the concept of phase given by \cite{Postle1981} to define the phase margin of a MIMO system in terms of the range of phases of the eigenvalues of a unitary perturbation (as in $\Delta$ in Figure \ref{fig:iqc_feedback}). This approach can be thought of as ``lumping'' the positive semidefinite portion of the perturbation's polar decomposition into the plant. \cite{Wang2008} treats a structured version of this problem, calculating the phase margin of an LTI system with respect to unitary diagonal perturbations. \cite{Chellaboina2008} extends this concept to perturbations with general block diagonal structures, while \cite{Srazhidinov2023} links the phase margin concept with the David-Wielandt shell to produce optimization problems guaranteed to converge to the exact multivariable gain and phase margin of a system. 

The principal limitation of the approaches taken in  \cite{Postle1981}-\cite{Srazhidinov2023} is that the classes of perturbations under consideration depend not only on their phase properties but also their gain properties. From a theoretical point of view, these attempts either do not define a scaling-invariant phase of a system, or fail to determine a scaling-invariant set of admissible perturbations. From a practical point of view, this dependency on gain limits the set of perturbations that may be considered, and increases the conservatism of the results obtained. 

Some efforts have  been made to  develop a notion of system phase based on the \emph{numerical range} of the transfer matrix. The numerical range of a matrix is a convex subset of the complex plane whose angular extent is invariant under congruence transformations, which may be thought of as a symmetric scaling of the system's inputs and outputs. The positive realness of a system's transfer function is equivalent to the requirement that the numerical range of its frequency response lies in the open right half plane, for all frequencies.  Numerical  ranges were used   by \cite{Tits1999} to  define the phase-sensitive structured singular value (PS-SSV).  This was  shown  to improve the conservatism of the structured singular value under the assumption that the numerical range of the perturbation is restricted to an angle bounded sector of the complex plane. \cite{Freudenberg1988} mentions the relationship between the range of phases contained in the numerical range and the angle and phase difference between the input and output vectors, but does not develop the connection further.

Recently,  a definition of the phases of a matrix utilising its numerical  range was given by \cite{WCKQ2020}, based on the so-called sectorial decomposition. These phases are invariant under congruence, in the same way that the singular values of a matrix are invariant under unitary transformations, and hence satisfies the desired scaling invariance property. A corresponding definition of system phase and a small phase theorem for MIMO LTI systems was introduced in \cite{Chen2024}. These definitions of phase are indeed invariant under scaling in the sense of congruence, and unlike previous ``small phase'' conditions, the resulting robustness set is defined only by its phases, and so is invariant under congruence. 

These new notions of system phase employing  matrix numerical range  have led to a range of novel results in robust stability analysis. Mixed gain- and phase-based stability criteria were introduced in \cite{Woolcock2023} and  \cite{Liang2024}.  Some comparisons of stability analysis methods for assessing  the small-signal stability of a wind farm connected to a series-compensated grid were investigated in \cite{Woolcock2024}; the gain/phase based approach was shown to determine stability with less conservatism than an alternative  Nyquist-based method.  In \cite{Wang2024b} phase-based methods were applied to  multi-agent systems  to study the synchronization problem of highly heterogeneous diverse agents. A summary of  recent work in the use  of phase-based  stability studies appeared in  \cite{Wang2024}. 

In  this paper, we further investigate system  stability using the notions of system phase  introduced in  \cite{Chen2024}.  We  aim to extend the robust stability result of \cite{Chen2024} by considering the structured stability problem. Given a block diagonal structure, we define a \emph{structured phase index} function $\psi$ on complex matrices analogous to the structured singular value $\mu$. We show that an upper bound for this function may be obtained by solving LMI optimization problems, and that a lower bound is given by the solution of a non-convex optimization problem. The main result of this paper shows that for a set of perturbations whose phases are bounded by this  phase  index  function,  the stability of the perturbed system is guaranteed, and that this phase-based criterion may be combined on a per-frequency basis with the structured singular value. This result enriches the robustness analysis of feedback interconnections by augmenting the existing $\mu$-based analysis with phase information on the perturbation, and augmenting the existing phase-based analysis of \cite{Chen2024} with structural information on the perturbation.
	
The paper is  organised as  follows: In Section \ref{sec:problem}, we formulate the structured robust stability problem for LTI systems; our task is to quantify the perturbations that yield a stable feedback  connection in terms of the matrix phase of their frequency response. Section \ref{sec:maths} provides a brief introduction to the matrix phase concepts used in the paper. Section \ref{sec:spi} defines a ``structured phase'' analogue to the structured singular value of a complex matrix, and provides an LMI-based method for calculating an upper bound of this quantity, as well as a non-convex optimization problem for computing a lower bound. Section \ref{sec:main} applies the results of Section \ref{sec:spi} to derive the main result of the paper on phase- and gain-bounded perturbations. Section \ref{sec:example} presents an example taken from \cite{Doyle1982} to show that the  criterion may be  less conservative than  purely  gain-based robust stability criteria. 

{\em Notation}: We denote by $\R$ and $\C$ the sets of real and complex numbers. For a complex number $z\in\C$, its phase is denoted by $\angle z$, and taken to be in the interval $(-\pi,\pi]$. $\C^{n\times m}$ denotes the set of $n\times m$ complex matrices. $\C^n$ denotes the set of $n$ dimensional complex column vectors. The $n\times n$ identity matrix is denoted $I_n$. For a matrix $A\in\C^{n\times m}$, we denote its conjugate transpose by $A^*$, and its induced 2-norm (i.e. maximum singular value) by $\|A\|$. If $A$ is square, we denote the inverse of its conjugate transpose by $A^{-*}$, its spectrum by $\sigma(A)$ and its determinant by $|A|$. $A$ has a decomposition given by:
\begin{align}
	A=\Re(A) + j\Im(A),
\end{align}
where $\Re(A)=(A+A^*)/2$ and $\Im(A)=(A-A^*)/2j$ are Hermitian matrices. If $A$ is Hermitian, $A>0$ ($A\geq 0$) means that $A$ is positive (semi-)definite, and $A<0$ ($A\leq 0$) means that $A$ is negative (semi-) definite. The notation $A>B$ ($A\geq B$) means that $A-B$ is positive (semi-)definite.

The 2-norm of a function $f  : [0,\infty) \rightarrow \R^n$ is denoted by:
\begin{align}
	\|f\| := \int_0^\infty f(t)^\T f(t) dt.
\end{align}
We then denote the space of finite-energy functions  as:
\begin{align}
	\L_2^n := \{f : [0,\infty) \rightarrow \R^n : ||f||<\infty\}.
\end{align}
$\L_{2e}^n$ denotes the extended function space:
\begin{align}
	\L_{2e}^n := \{f : [0,\infty)\rightarrow \R^n \: : \: ||P_T f|| < \infty,  \forall T\geq 0\}.
\end{align}
where $P_T$ denotes the truncation operator for any  $T \geq  0$. 

$\RLi$ denotes the space of rational functions with no poles on the imaginary axis. $\RHi$ denotes the subspace with no poles in the RHP. $\RLi^{n\times m}$ and $\RHi^{n\times m}$ denote the spaces of $n\times m$ matrices of functions in $\RLi$ and $\RHi$ respectively. The Fourier transform of a function $f\in\L_2^n$ is denoted $\hat{f}$.

A \emph{sector} of the complex plane is denoted:
\begin{align}
	\mathbf{sec}(\alpha, \beta) := \{\rho e^{j\theta}\: : \: \rho\geq 0, \theta\in[\alpha, \beta]\}.
\end{align}
for some bounding angles $\alpha,\beta\in\R$. The convex hull of a set $X \subset \C^n$ is denoted $\operatorname{Co}(X)$.

\section{Problem Formulation} \label{sec:problem}
We now formally define the robust structured stability problem depicted in Figure \ref{fig:iqc_feedback}. Recall that $G,\Delta\in\RHi^{n\times n}$ are stable, causal LTI systems. They define the following system of equations:
\begin{align}
	e_1 &= -\Delta e_2 + u_1,\\
	e_2 &= G e_1 + u_2,
\end{align}
where $u_1,u_2\in\L_2^n$ and $e_1,e_2\in\L_{2e}^n$. The feedback interconnection of $(G,\Delta)$ is  \emph{well-posed} if, for each pair of inputs $u_1,u_2\in\L_2^n$, there are unique solutions $e_1,e_2\in\L_{2e}^n$. The interconnection is  \emph{stable} if $e_1,e_2\in\L_{2}^n$  for each  $u_1,u_2\in\L_2^n$.  The interconnection is stable if and only if the transfer function $(I+G\Delta)^{-1}\in\RHi^{n\times n}$. 	Given a plant $G$, a perturbation $\Delta$ is  \emph{admissible} with respect to $G$ if $(G,\Delta)$ is a well-posed and stable interconnection.

We will now introduce the notation needed to define a block-diagonal structured set of  perturbations, closely following \cite{zhou1996}. We denote a pair of tuples of positive integers of length $s$ and $b$ respectively, where $s,b\geq 0$ and $s+b>0$:
\begin{align}
	\mathbf{n} &= (n_1, n_2,\dots, n_s), \label{Sdimension} \\
	\mathbf{m} &= (m_1, m_2,\dots, m_b) \label{Bdimension}.
\end{align}
We refer to $\mathbf{n}$ and $\mathbf{m}$ as \emph{block dimension tuples compatible with $G$} if $\sum\mathbf{n}+\sum\mathbf{m}=n$. For the sake of compactness, we will denote such a choice of compatible block dimensions as $\chi=(\mathbf{n}, \mathbf{m})$.  A choice of $\chi$ then defines a set of \emph{structured matrices} as follows:
\begin{align} \label{eq:problem:Bmatrix}
	\begin{split}
		\mathbf{B}_\chi := \{\mathrm{blkdiag}(b_1 I_{n_1}, b_2 I_{n_2}, \dots, b_s I_{n_s},\\
		B_1, B_2,\dots, B_b) \: : \: b_k\in\C\text{ and }B_k\in\C^{m_k\times m_k}\\ \text{ for }k=1,\dots,s\text{ and } i=1,\dots,b\}.
	\end{split}
\end{align}
The first $s$ blocks of the form $b_i I_{n_i}$ are called \emph{scalar blocks}, while the remaining $b$ blocks are called \emph{full blocks}. We then define the set of \emph{structured perturbations} whose frequency response belongs to $\mathbf{B}_\chi$.
\begin{definition}
	For a given choice of block dimensions $\chi$ compatible with $G$, we define a \emph{set of structured perturbations}:
	\begin{align} \label{eq:metric:Delta}
		\mathbf{\Delta}_\chi := \{\Delta\in\RHi^{n\times n} \: : \: \Delta(j\omega)\in\mathbf{B}_\chi  \,\forall \omega\in[0,\infty]\}.
	\end{align}
\end{definition}
The  {\it  structured robust stability problem}   involves finding the largest set of admissible perturbations contained in this structured set according to some metric. The structured singular value introduced by  \cite{Doyle1982} identifies such a set according to a frequency-dependent gain bound.  Our Theorem \ref{thm:spr:phase_stability}  addresses  the analogous phase-bounded problem:

\begin{problem} \label{Prob1}
	Given a plant $G$ and a set of structured perturbations, introduce a frequency-dependent phase bound to identify the admissible  perturbations.
\end{problem}

A  special  case of Problem  \ref{Prob1} is the case  of  unstructured perturbations, for which   $\mathbf{\Delta}=\RHi^{n\times n}$. This was addressed in  the small phase theorem of \cite{Chen2024}. Thus the first result of this paper can be thought of as extending the small phase theorem to the general case of structured perturbations.

Our Theorem \ref{thm:spr:mixed_stability}  will provide less conservative solutions for  the  structured robust stability problem problem  by  addressing  the combined  gain and phase bounded  problem:

\begin{problem} \label{Prob2}
	Given a plant $G$ and a set of structured perturbations, provide  
	mixed gain  and  phase criteria  to identify  the admissible  perturbations.	
\end{problem}

\section{Mathematical Preliminaries} \label{sec:maths}
In this section  we introduce  the results on matrix phase that will  be  needed for  our robustness  analysis. The sources are \cite{WCKQ2020}, \cite{Chen2024}  and \cite{Horn1991}. 

\begin{definition}
	The \emph{numerical range} of a matrix $A\in\C^{n\times n}$ is:
	\begin{align}
		W(A) := \{ x^*Ax \: | \: x\in\C^{n}, \, \|x\|=1 \}.
	\end{align}
\end{definition}
Some key properties of the numerical range are:
\begin{enumerate}[label=(W\arabic*)]
	\item $W(A)$ is a convex and compact set in $\C$. \label{prop:phase:Wconvex}
	\item $\sigma(A)\subseteq W(A)$. \label{prop:phase:Wspectrum}
	\item $W(cA)=cW(A)$ for $c\in\C$. \label{prop:phase:Wscalar}
	\item $W(\Re(A)) = \Re(W(A))$. \label{prop:phase:Wreal}
	\item $W(\Re(A)) = \Im(W(A))$. \label{prop:phase:Wimag}
	\item $W(\blkdiag(A_1,A_2)) = \operatorname{Co}(W(A_1)\cup W(A_2))$. \label{prop:phase:Wblkdiag}
\end{enumerate}

If we consider the vector $v\in\C^n$ the input and $Av$ the output of $A$, the range of phases contained within $W(A)$ can be thought of as encoding the range of possible phases of the inner product of the input and output, which will provide the basis of the definition of the phases of $A$. This is analogous to the relationship between the singular values of a matrix and the ratio of the 2-norm of its input and output.

\begin{definition} \label{defn:phase:anr}
	The \emph{angular numerical range} of a matrix $A\in\C^{n\times n}$ is:
	\begin{align}
		W'(A) := \{ x^*Ax \: | \: x\in\C^{n}, \, \|x\|\neq 0 \}.
	\end{align}
\end{definition}
It can be seen that $W'(A)$ is a convex cone, i.e. $\alpha x + \beta y \in W'(A)$ for any $\alpha,\beta\geq 0$ and $x,y\in W'(A)$. Specifically, it is the smallest convex cone that contains $W(A)$.
\begin{definition}
	The \emph{field angle} of a matrix $A\in\C^{n\times n}$ is denoted $\Theta(A)$. If $0$ is contained in the interior of $W(A)$, $\Theta(A)$ is defined to be $2\pi$. Otherwise, $\Theta(A)$ is the minimum angle between a pair of rays from the origin that subtend $W(A)$.
\end{definition}

\begin{figure}
	\centering
	\subcaptionbox{}
	{\includegraphics[scale=0.8]{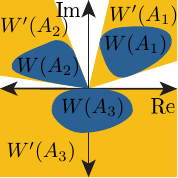}}
	\subcaptionbox{}
	{\includegraphics[scale=0.8]{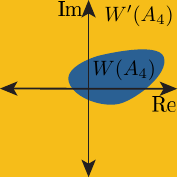}}
	\caption{$A_1$ is sectorial, $A_2$ is quasi-sectorial but not sectorial, $A_3$ is semi-sectorial but not quasi-sectorial, and $A_4$ is non-semi-sectorial.}
	\label{fig:phase:sectoriality}
\end{figure}
Clearly, if $0\notin W(A)$, $\Theta(A)<\pi$ by the convexity of $W(A)$. Such a matrix is known as a \emph{sectorial} matrix. More generally, we categorise matrices by their field angle as follows:
\begin{enumerate}
	\item If $\Theta(A)<\pi$, then $A$ is \emph{quasi-sectorial.} All sectorial matrices are quasi-sectorial, but quasi-sectorial matrices may contain $0$ in their numerical range at a sharp point on its boundary. This case is illustrated by $A_2$ in Figure \ref{fig:phase:sectoriality}.
	\item If $\Theta(A)\leq\pi$, then $A$ is \emph{semi-sectorial.} Semi-sectorial matrices may contain $0$ in the boundary of their numerical range, not necessarily at a sharp point. This case is illustrated by $A_3$ in Figure \ref{fig:phase:sectoriality}.
	\item If $\Theta(A)=2\pi$, then $A$ is \emph{non-semi-sectorial.} A matrix is non-semi-sectorial if and only if $0$ is in the interior of its numerical range. This case is illustrated by $A_4$ in Figure \ref{fig:phase:sectoriality}.
\end{enumerate}

\begin{remark}
	The geometric interpretation of $\Theta(A)$ is that $W(A)\subseteq\mathbf{sec}(\alpha,\beta)$ with $\beta-\alpha=\Theta(A)$. This means that a semi-sectorial matrix has a numerical range lying in a  convex sector of the complex plane.
\end{remark}

A sectorial matrix has a \emph{sectorial decomposition} as follows:
\begin{align}
	A=T^* D T,
\end{align}
where $D,  T\in\C^{n\times n}$ are invertible and $D=\diag(d_1,\dots,d_n)$ with $|d_i|=1$ for $i=1,\dots,n$. From Definition \ref{defn:phase:anr}, we obtain
\begin{align}
	W'(A) = \{k_1 d_1 + \dots + k_n d_n\: : \: k_1,\dots,k_n \geq 0 \}\setminus \{0\}.
\end{align}
Therefore, the phases of the diagonal elements determine the angular extent of $W'(A)$ and $W(A)$. 
\begin{definition}[Matrix phase\cite{Chen2024}]
	The \emph{matrix phases} of a sectorial matrix $A\in\C^{n\times n}$ are the $n$ real numbers denoted:
	\begin{align}
		\ophi(A) = \phi_1(A) \geq \phi_2(A) \geq \dots \phi_n(A) = \uphi(A),
	\end{align}
	chosen modulo $2\pi$ so that the \emph{phase center} $\gamma(A)=(\ophi(A) + \uphi(A))/2$ satisfies $\gamma(A)\in(-\pi,\pi]$, and:
	\begin{align}
		A = T^* \diag(e^{j\phi_1(A)}, e^{j\phi_2(A)}, \dots, e^{j\phi_n(A)}) T
	\end{align}
	is a sectorial decomposition of $A$ for some invertible $T\in\C^{n\times n}$.
\end{definition}

This definition may be extended to quasi-sectorial matrices by noting that if $0$ is an eigenvalue of a quasi-sectorial matrix $A\in\C^{n\times n}$, it must be a normal eigenvalue \cite{Horn1991}.  A  quasi-sectorial matrix  $A$ permits a decomposition of the form \cite{Chen2024}:
\begin{align} 
	A = U^* \bmat 0 & 0 \\ 0 & \hat{A} \emat U, \label{eq:phase:quasi}
\end{align}
where $\hat{A}\in\C^{m\times m}$ is a sectorial matrix with $m\leq n$ and $U\in\C^{n\times n}$ is unitary. Then $A$ is defined to have $m$ phases, which are the phases of $\hat{A}$. The sectorial decomposition can be further generalised to semi-sectorial matrices, with the details given in \cite{Chen2024}.

\begin{figure}
	\centering
	\includegraphics[scale=0.9]{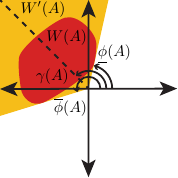}
	\caption{Relationship between $\uphi(A)$, $\ophi(A)$ and $W(A)$.}
	\label{fig:phase:matphase_illust}
\end{figure}
The phases $\ophi(A)$ and $\uphi(A)$ may be interpreted geometrically as the angles of the most anticlockwise and clockwise rays from the origin subtending $W(A)$, while $\gamma(A)$ is the angle of the ray bisecting the sector $W'(A)$, as depicted in Figure \ref{fig:phase:matphase_illust}. We also have $\Theta(A)=\ophi(A)-\uphi(A)$.
This definition of matrix phase can be used to bound the phases of the eigenvalues of a product of two matrices, which translates directly into the phase-based stability results given in \cite{Chen2024}.
\begin{lemma} \label{lemma:phase:matrix_spt}
	Suppose $A,B\in\C^{n\times n}$ are semi-sectorial matrices. If $\lambda\in\sigma(AB)$, then:
	\begin{align}
		\uphi(A) + \uphi(B) \leq \angle\lambda \leq \ophi(A) + \ophi(B),
	\end{align}
	up to a multiple of $2\pi$.
\end{lemma}
A direct corollary is that if:
\begin{align}
	\ophi(A) + \ophi(B) &< \pi,\\
	\uphi(A) + \uphi(B) &> -\pi,
\end{align}
we may guarantee:
\begin{align}
	|I+AB|\neq 0.
\end{align}

We collect a number of simple properties of matrix phases here.
\begin{enumerate}[label=(P\arabic*)]
	\item $\phi_i(zA)=\angle{z}+\phi_i(A)$ for any nonzero $z\in\C$ with $\angle{z}\in(-\pi-\gamma(A), \pi-\gamma(A)]$. \label{prop:phase:Pscalar}
	\item $\Theta(zA)=\Theta(A)$ for any nonzero $z\in\C$.
	\item A (quasi-)sectorial matrix is positive (semi-)definite if and only if all its phases are zero.
	\item If $A$ is semi-sectorial and invertible, $\ophi(A^{-1})=-\uphi(A)$ and $\uphi(A^{-1})=-\ophi(A)$. \label{prop:phase:Pinv}
	\item If $X_1,\dots,X_k$ are semi-sectorial matrices with phases lying in $[\alpha,\beta]$, where $\beta-\alpha\leq\pi$, then $X=\blkdiag(X_1,\dots,X_k)$ is a semi-sectorial matrix with $\uphi(X)\geq\alpha$ and $\ophi(X)\leq\beta$. \label{prop:phase:Pblkdiag}
\end{enumerate}

\section{Structured Phase Index} \label{sec:spi}
Motivated by the problem of defining a frequency-dependent phase bound on the set of admissible perturbations, we consider sets of structured matrices $\mathbf{B}_\chi$ as defined in \eqref{eq:problem:Bmatrix}.  We recall the definition of the structured singular value.
\begin{definition}[Structured singular value \cite{zhou1996}] \label{def:mu}
	The structured singular value of a matrix $A\in\C^{n\times n}$ with respect to a set of structured matrices $\mathbf{B}_\chi$ as defined in \eqref{eq:problem:Bmatrix} is defined as:
	\begin{align}
		\mu_\chi(A) := \frac{1}{\min\{\|B\| : |I+A B| = 0,\, B\in\mathbf{B}_\chi\}}
	\end{align}
\end{definition}
An upper bound for $\mu$ may be calculated by the use of so-called \emph{$D$-scaling} matrices belonging to the following set:
\begin{align}	\begin{split}
		\mathbf{D}_\chi := \{\mathrm{blkdiag}(D_1,D_2,\dots,D_s,\\d_1 I_{m_1},d_2 I_{m_2},\dots, d_{b}I_{m_b})\: : \: D_k\in\C^{n_k\times n_k}\text{ and } d_k\in\C\ \\ \text{ for } i=1,\dots,s\text{ and } k=1,\dots,b\}.
	\end{split} \label{eq:problem:D_structure}
\end{align}
For each scalar block of $B\in\mathbf{B}_\chi$, the matrix  $D\in\mathbf{D}_\chi$  has a full block in the corresponding position, and vice versa, and so $DB=B D$. Then, the following is an upper bound for $\mu$\cite{Doyle1982}:
\begin{align} \label{eq:problem:mu_bound}
	\mu_\chi(A)\leq \overline{\mu}_\chi(A) := \inf_{D\in\mathbf{D}_\chi, |D|\neq 0} \|D A D^{-1}\|.
\end{align}
An optimal $D$-scaling matrix may be found via LMI (and other methods,) as implemented in MATLAB's \texttt{mussv} routine.

Our aim in this section will  be to  develop an  analogous ``structured phase'' measure on complex matrices.  We then provide upper and lower bounds on this quantity, the former of which can be computed using LMI methods. This definition will be used in Section \ref{sec:main} to obtain our main results on the robust stability of feedback interconnections.

\begin{figure}
	\centering
	\includegraphics{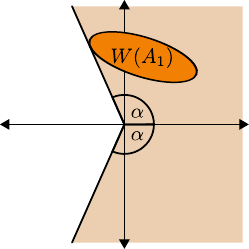}
	\caption{A sectorial matrix $A_1$ with phase index $\alpha$.}
	\label{fig:phase:radius}
\end{figure}

We  begin  by introducing the notion of a \emph{phase index}, intended as a phase analogue to the largest singular value of a matrix.
\begin{definition}
	The \emph{phase index} of a matrix $A\in\C^{n\times n}$ is defined as follows: 
	\begin{align}
		\phi(A) := \sup_{z\in W(A)} |\angle z|.
	\end{align}
\end{definition}
A matrix with phase index $\alpha\in[0,\pi]$ has numerical range $W(A)$ that lies within $\mathbf{sec}(-\alpha,\alpha)$, as depicted in Figure \ref{fig:phase:radius}. It follows that if the negative real line is contained in $W'(A)$, $\phi(A)=\pi$. This is the case when $A$ is non-semi-sectorial, but is also the case when $A$ is semi-sectorial and $\pi$ or $-\pi\in[\uphi(A), \ophi(A)]$. More precisely, if $A$ is semi-sectorial, then:
\begin{align}
	\phi(A) = \min(\pi, \max(|\uphi(A)|, |\ophi(A)|)).
\end{align}
It then follows from Lemma \ref{lemma:phase:matrix_spt} that if:
\begin{align}
	\phi(A) + \phi(B) < \pi, \label{eq:phase:pindex_spt}
\end{align}
for some $A,B\in\C^{n\times n}$, we have $|I+AB|\neq 0$. We also have the following lemma:
\begin{lemma} \label{lemma:spi:DeltaD}
	Suppose $B\in\mathbf{B}_\chi$ and $D\in\mathbf{D}_\chi$ satisfy $\Theta(B)+\Theta(D)\leq \pi/2$. Then:
	\begin{align}
		\phi(BD) \leq \phi(B)+\phi(D).
	\end{align}
\end{lemma}
\begin{proof}
	If $B$ has block diagonal structure as notated in \eqref{eq:problem:Bmatrix} and $D$ as in \eqref{eq:problem:D_structure}, then we can write:
	\begin{align}
		B D = \blkdiag(\delta_1 D_1,\dots \delta_s D_s, d_1\Delta_1, d_b \Delta_b).
	\end{align}
	By \ref{prop:phase:Pblkdiag} and \ref{prop:phase:Pscalar}, each block's phases lie in the  interval  $[\uphi(B)+\uphi(D), \ophi(B)+\ophi(D)]$. Then, \ref{prop:phase:Pblkdiag} implies that $BD$ is a semi-sectorial matrix with phases in that interval, from which the result follows.
\end{proof}
We now define the \emph{structured phase index},  in analogy to the structured singular value. 
\begin{definition} \label{defn:spr}
	Given a matrix $A \in \C^{n\times n}$, its \emph{structured phase index} with respect to a compatible block dimension $\chi$ is:
	\begin{align} \label{eq:SPM_def}
			\psi_\chi(A) := \pi - \inf \{\phi(B) \: :\:|I+AB|= 0,\,B\in\mathbf{B}_\chi \}
	\end{align}
\end{definition}

There is no straightforward way to calculate $\psi$. Clearly, it follows from \eqref{eq:phase:pindex_spt} that 
\begin{align} \label{eq:upp_bound}
	\psi_\chi(A)\leq\phi(A).
\end{align}
However, this offers no more information than the small phase theorem itself, and can be an arbitrarily conservative estimate of the true value of $\psi$. We can address this by using $D$-scaling matrices as multipliers.
\begin{lemma} \label{lemma:spi:D_multipliers}
	For any $A\in\C^{n\times n}$, $B\in\mathbf{B}_\chi$ and invertible $D\in\mathbf{D}_\chi$, if:
	\begin{align}
		\phi(B) &< \pi - \phi(AD) - \phi(D), \label{eq:spi:phi_req} \\
		\Theta(B) &\leq \pi - \Theta(D), \label{eq:spi:Theta_req}
	\end{align}
	are satisfied, we have $|I+AB|\neq 0$.
\end{lemma}
\begin{proof}
	Because \eqref{eq:spi:Theta_req} is satisfied, we may invoke Lemma \ref{lemma:spi:DeltaD} and \ref{prop:phase:Pinv} to state $\phi(D^{-1}B)\leq \phi(D)+\phi(B)$. From \eqref{eq:spi:phi_req} we also have $\phi(AD)<\pi-\phi(B)-\phi(D)$. By applying Lemma \ref{lemma:phase:matrix_spt} to $AD$ and $D^{-1}B$, the conclusion follows.
\end{proof}

A corollary of this result is that if Lemma \ref{lemma:spi:D_multipliers} is satisfied by some particular choice of $D$, it is also satisfied by $\tilde{D}=e^{-j\gamma(D)}D$, for which:
\begin{align}
	\phi(\tilde{D}) = \Theta(\tilde{D})/2.
\end{align}
It follows that when searching for $D$-scaling matrices satisfying Lemma \ref{lemma:spi:D_multipliers}, we need only search over $D$-scaling matrices satisfying $\phi(D)<\pi/2$, i.e. $\Re(D)>0$.

In order to prove $\psi_\chi(A)<\alpha$, we need to prove that $|I+AB|\neq 0$ for any $B\in\mathbf{B}_\chi$ satisfying only $\phi(B)\leq \pi-\alpha$. A consequence of this is that if we wish to prove this for $\alpha\leq\pi/2$, we require $\Theta(D)=0$. The previous corollary implies that we need only search for a positive definite $D$ in this case that satisfies the requirement:
\begin{align}
	\phi(AD) < \alpha.
\end{align}
Therefore, if we search for a positive definite $D$ satisfying Lemma \ref{lemma:spi:D_multipliers}, and fail to find one, it follows that we cannot use $D$ multipliers to prove that $\alpha\leq\pi/2$. More generally, we require $\Theta(D)\leq\max(0, 2\alpha-\pi)$, and we need only search over $D$ satisfying:
\begin{align}
	\phi(D) &\leq \max(0, \alpha - \pi/2),  \label{eq:spi:D_constraint}\\
	\phi(AD) + \phi(D) &< \alpha. \label{eq:spi:AD_constraint}
\end{align}
This leads to the definition of an upper bound on $\psi$ as follows.
\begin{definition}
	\begin{align}
		\begin{split}
			\overline{\psi}_\chi(A) := \inf\{\alpha\in[0,\pi)\: :\: \exists D\in\mathbf{D}_\chi \text{ with } |D|\neq 0\\ \text{ satisfying \eqref{eq:spi:D_constraint}-\eqref{eq:spi:AD_constraint} for }\alpha\}. 
		\end{split} \label{eq:spi:opsi}
	\end{align}
\end{definition}
It is clear from Lemma \ref{lemma:spi:D_multipliers} that this is indeed an upper bound for $\psi$, i.e. that:
\begin{align}
	\overline{\psi}(A) + \phi(B) <\pi,
\end{align}
implies $|I+AB|\neq 0$. In the following, we will show that it may be calculated efficiently using LMI optimization problems. We will use the following technical lemma.
\begin{lemma} \label{lemma:phase:sector_lmi}
	A matrix $A\in\C^{n\times n}$ is quasi-sectorial with:
	\begin{align}
		\phi(A)\leq\arctan\kappa, \label{eq:phase:sector_lmi_pindex}
	\end{align}
	for some $\kappa\geq 0$ if and only if the following matrix inequalities are satisfied:
	\begin{align}
		\kappa\Re(A) &\geq \Im(A), \label{eq:phase:sector_lmi:lmi1} \\
		\kappa\Re(A) &\geq -\Im(A). \label{eq:phase:sector_lmi:lmi2}
	\end{align}
\end{lemma}
\begin{proof}
	(Sufficiency:) The matrix inequalities taken together imply $\Re(A)\geq 0$, which by \ref{prop:phase:Wreal} implies that $W(A)$ lies in the closed RHP, and is therefore semi-sectorial. The inequality \eqref{eq:phase:sector_lmi:lmi1} may be rearranged to show:
	\begin{align}
		\Re((\kappa-j)A) \geq 0.
	\end{align}
	By \ref{prop:phase:Pscalar}, this implies that $\uphi(A)\geq -\arctan\kappa$. Similarly, \eqref{eq:phase:sector_lmi:lmi2} may be rearranged to show $\ophi(A)\leq\arctan\kappa$, showing $\phi(A) \leq \arctan\kappa$ as needed.
	
	(Necessity:) The preceding steps may be reversed to show that the phase constraint \eqref{eq:phase:sector_lmi_pindex} imply the LMIs \eqref{eq:phase:sector_lmi:lmi1}-\eqref{eq:phase:sector_lmi:lmi2}.
\end{proof}

The following LMI problem allows the first term of the minimum in \eqref{eq:spi:upsi} to be determined in the case that it is no greater than $\pi/2$.
\begin{lemma} \label{lemma:spr:lmi_pdef}
	Given a matrix $A\in\C^{n\times n}$,
	\begin{align}
		\overline{\psi}_\chi(A)\leq\arctan\kappa \label{lemma:spi:lmi_pdef_bound}
	\end{align}
	for some $\kappa>0$ if and only if:
	\begin{align}
		D &> 0\label{eq:spr:lmi_pdef_1}\\
		\kappa\Re(AD) &\geq \Im(AD) \label{eq:spr:lmi_pdef_2}\\
		\kappa\Re(AD) &\geq -\Im(AD) \label{eq:spr:lmi_pdef_3}
	\end{align}
	for some $D^*=D\in\mathbf{D}_\chi$.
\end{lemma}
\begin{proof}
	We will proceed by showing that these LMI conditions are equivalent to \eqref{eq:spi:D_constraint}-\eqref{eq:spi:AD_constraint} for all $\alpha>\arctan\kappa$, placing an upper bound on the infimum in \eqref{eq:spi:opsi}. Firstly, $D>0$ if and only if $\phi(D)=0$ and $D$ is invertible, so the invertibility requirement and \eqref{eq:spi:D_constraint} are equivalent to \eqref{eq:spr:lmi_pdef_1} for any $\alpha$. Lemma \ref{lemma:phase:sector_lmi} then implies that \eqref{eq:spr:lmi_pdef_2}-\eqref{eq:spr:lmi_pdef_3} are equivalent to the condition $\phi(AD)\leq \arctan\kappa$, and so \eqref{eq:spi:AD_constraint} is satisfied for all $\alpha>\arctan\kappa$, as required.
\end{proof}

\begin{remark}
The problem of finding the smallest $\kappa$ for which Lemma \ref{lemma:spr:lmi_pdef} is feasible is a type of LMI optimization problem called the \emph{generalized eigenvalue problem.} Such problems are quasi-convex, and solvers exist that are guaranteed to converge on the infimum for $\kappa$ in polynomial time. In the case where $A$ is nonsingular, MATLAB's \texttt{gevp} provides an efficient solution to the problem. If $A$ is singular, numerical issues arise because \eqref{eq:spr:lmi_pdef_2}-\eqref{eq:spr:lmi_pdef_3} cannot be strictly satisfied. In this case, an approach using the \texttt{feasp} solver with a bisection method to estimate the minimum $\kappa$ is required.
\end{remark}

\begin{lemma} \label{lemma:spr:lmi_accr}
	Given a matrix $A\in\C^{n\times n}$, 
	\begin{align}
		\overline{\psi}_\chi(A) \leq \pi/2 + \arctan\kappa
	\end{align}
	for some $\kappa>0$ if and only if:
	\begin{align}
		\Re(D) & > 0 \label{eq:spi:lmi_accr_1} \\ 
		\Re(AD) &\geq 0 \label{eq:spi:lmi_accr_2} \\
		\kappa\Re(D) &\geq \Im(D) \label{eq:spi:lmi_accr_3}\\
		\kappa\Re(D) &\geq -\Im(D) \label{eq:spi:lmi_accr_4}
	\end{align}
	for some $D\in\mathbf{D}_\chi$.
\end{lemma}
\begin{proof}
	\eqref{eq:spi:lmi_accr_3}-\eqref{eq:spi:lmi_accr_4} are by Lemma \ref{lemma:phase:sector_lmi} equivalent to $\phi(D)\leq \arctan\kappa$, and so are equivalent to \eqref{eq:spi:D_constraint} for $\alpha\geq\pi/2+\arctan\kappa$. Given these phase constraints, \eqref{eq:spi:lmi_accr_1} is then equivalent to the invertibility of $D$. \eqref{eq:spi:lmi_accr_4} then implies $\phi(AD)\leq\pi/2$, which is sufficient to prove \eqref{eq:spi:AD_constraint} for $\alpha>\pi/2+\arctan\kappa$.
\end{proof}

Algorithm \ref{alg1} summarizes the two-step procedure for calculating $\overline\psi_\chi(A)$. 
\begin{algorithm}[h!] 
	\caption{Algorithm for calculating $\overline\psi_\chi$}
	\label{alg1}
	\begin{algorithmic}
		\Require $A\in\C^{n\times n}$, $\chi$.
		
		\If{Lemma \ref{lemma:spr:lmi_pdef} is feasible}
		$\overline\psi_\chi\gets \arctan\kappa$
		\ElsIf{Lemma \ref{lemma:spr:lmi_accr} is feasible}
		$\overline\psi_\chi\gets \pi/2 + \arctan\kappa$
		\Else \ 
		$\overline\psi_\chi\gets\pi$
		\EndIf
	\end{algorithmic}
\end{algorithm}

To complement the upper bound $\overline\psi_\chi$, we desire a lower bound. This does not provide a sufficient condition for stability, but does allow a bound to be placed on the conservatism of $\overline\psi_\chi$. If we consider $B=\frac{1}{\lambda}I_n$, where $\lambda\in\sigma(A)$, the largest absolute phase of $A$'s eigenvalues gives a lower bound:
\begin{align}    \label{eq:low_bound}
	\psi_\chi(A)\geq\max_{\lambda\in\sigma(A)} |\angle\lambda|
\end{align}
This lower bound can generally be arbitrarily optimistic. To address this, we introduce the following:
\begin{align} \label{eq:spi:upsi}
	\underline\psi_\chi(A) := \max_{X\in\mathbf{B}_\chi} \{|\angle\lambda| \: : \: \lambda\in\sigma(XAX^*)\}.
\end{align}
To see that this is indeed a lower bound, take $B=\lambda^{-1}X^*X$, which satisfies $\phi(B)\leq\underline\psi_\chi(A)$.

The optimization problem contained in \eqref{eq:spi:upsi} is non-convex, and can have multiple local maxima. A detailed study of the properties of this optimization problem is outside of the scope of this paper. We will present an objective function and gradient that can be used with numerical optimization routines to obtain a potentially non-optimal lower bound. We parameterize the scaling matrix and the scaled version of $A$ as:
\begin{align}
	X(x) &:= x_1 X_1 + \dots x_k X_k,\\
	\tilde{A}(x) &:= X(x)AX(x),
\end{align}
where $x=\bmat x_1&\dots&x_k\emat^\T\in\R^k$ and $X_1,\dots,X_k$ form a basis for all Hermitian matrices in $\mathbf{B}_\chi$. We first define:
\begin{align}
	\lambda_\star(x) := \operatorname*{argmax}_{\lambda\in\sigma(\tilde{A}(x))} |\angle\lambda|.
\end{align}
The objective function to maximize is then:
\begin{align} \label{eq:spi:upsi_obj}
	f(x) := \max_{x\in\R^k} |\angle\lambda_\star(x)|.
\end{align}
We will now find an expression for the gradient of $f$ in the case that $\lambda_\star(x)$ is unique and a simple eigenvalue of $\tilde{A}(x)$. In this case, we may apply Theorem \ref{thm:maths:d_eig} to obtain the partial derivative of $\lambda_\star$ with respect to each element of $x$:
\begin{align} \label{eq:grad_partials}
	\frac{\partial\lambda_\star}{\partial x_i} &= \frac{v(x)^* \tilde{A}(x) (X(x)X_i + X_i X(x)) u(x)} {v(x)^* u(x)},
\end{align}
for $i=1,\dots,k$, where $u(x),v(x)\in\C^{n}$ are the right and left eigenvectors of $\tilde{A}(x)$ associated with $\lambda_\star(x)$. We may write down the following from the partial derivatives of $\operatorname{arctan2}$:
\begin{align}
	\frac{\partial f}{\partial \Re\lambda_\star} &= -\operatorname{sgn}(|\angle\lambda_\star|) \frac{\Im(\lambda_\star)}{|\lambda_\star|^2},\\
	\frac{\partial f}{\partial \Im\lambda_\star} &= +\operatorname{sgn}(|\angle\lambda_\star|) \frac{\Re(\lambda_\star)}{|\lambda_\star|^2}.
\end{align}
Then, each component of the gradient may be written:
\begin{align} \label{eq:spi:upsi_grad}
	\frac{\partial f}{\partial x_i} &= \frac{\partial f}{\partial \Re\lambda_\star}\Re \frac{\partial \lambda_\star}{\partial x_i} + \frac{\partial f}{\partial \Im\lambda_\star} \Im\frac{\partial \lambda_\star}{\partial x_i},
\end{align}
for $i=1,\dots,k$.

MATLAB's \texttt{fminunc} may  be used to obtain the estimates of $\underline\psi$, with the objective function $f$ and its gradient in \eqref{eq:grad_partials},  and using $X_c=I_n$ as the starting point and \texttt{trust-region} as the algorithm.

\section{Main Results} \label{sec:main}

We are now able to  present  our the main  results of   on  the stability of LTI  feedback interconnected systems subject  to structured perturbations. Our main  result shows  that if the frequency response of the plant and perturbation respect  the phase bound  at  all  frequencies, an IQC can be formulated that guarantees the stability of the interconnection, and hence the admissibility of the perturbation. 

Our first result provides our solution to Problem 1: 

\begin{theorem} \label{thm:spr:phase_stability}
	Consider a well-posed feedback system $(G, \Delta)$, where $G, \Delta \in\RHi^{n\times n}$. Let $\chi=(\mathbf{n},	\mathbf{m})$ be a pair of block dimension tuples compatible with $G$. If $\Delta\in\mathbf{\Delta}_\chi$ and:
	\begin{align} \label{eq:spr:system}
		\phi(\Delta(j\omega))< \pi - \overline{\psi}(G(j\omega)),
	\end{align}
	for all $\omega\in[0,\infty]$, then the interconnection $(G,\Delta)$ is stable.
\end{theorem}
\begin{proof}
	Define the function $D:[0,\infty]\rightarrow \C^{n\times n}$, where $D(\omega)$ is equal to some matrix in $\mathbf{D}_\chi$ satisfying the condition of the supremum in \eqref{eq:spi:opsi} with $A=G(j\omega)$ and $\alpha=\alpha(\omega)$ satisfying:
	\begin{align}
		\phi(\Delta(j\omega)) \leq \alpha(\omega) < \pi - \overline\psi(G(j\omega)).
	\end{align}
	It follows from \eqref{eq:spi:D_constraint}-\eqref{eq:spi:AD_constraint} that:
	\begin{align}
		\phi(G(j\omega)D(\omega)) + \phi(D(\omega)) < \pi - \phi(\Delta(j\omega)),
	\end{align}
	and then by Lemma \ref{lemma:spi:DeltaD} that:
	\begin{align}
		\phi(G(j\omega)D(\omega)) + \phi(D(\omega)^{-1}\Delta(j\omega)) < \pi.
	\end{align}
	It then follows from \ref{prop:phase:Pscalar} that there exists $\beta(\omega)\in(-\pi/2,\pi/2)$ so that:
	\begin{align}
		\phi(e^{j\beta(\omega)}G(j\omega)D(\omega)) < \pi/2, \label{eq:result:phase_GD}\\
		\phi(e^{-j\beta(\omega)}D(\omega)^{-1}\Delta(j\omega)) \leq \pi/2. \label{eq:result:phase_DDelta}
	\end{align}
	Then define the multiplier function:
	\begin{align}
		\Pi(j\omega) &= \bmat 0 & -e^{-j\beta(\omega)}D(\omega)^{-1} \\ -e^{j\beta(\omega)}D(\omega)^{-*} & 0 \emat
	\end{align}
	$\Delta$ satisfies the IQC defined by $\Pi$ if and only if the frequency domain inequality (FDI) is satisfied:
	\begin{align}
		\bmat I \\ -\Delta(j\omega) \emat^* \Pi(j\omega) \bmat I \\ -\Delta(j\omega) \emat &\geq 0, \\
		\iff \Re(e^{-j\beta(\omega)}D(\omega)^{-1}\Delta(j\omega)) &\geq 0
	\end{align}
	which is indeed guaranteed by \eqref{eq:result:phase_DDelta}. The corresponding FDI for $G$ is:
	\begin{align}
		\bmat G(j\omega) \\ I \emat^* \Pi(j\omega) \bmat G(j\omega) \\ I \emat &\leq -\epsilon G(j\omega)^*G(j\omega), \\
		\Re(e^{j\beta(\omega)}G(j\omega)D(\omega)) &\geq -\epsilon D(\omega)^*G(j\omega)^* G(j\omega)D(\omega).
	\end{align}
	Note that by \eqref{eq:result:phase_GD}, $G(j\omega)D(\omega)$ is quasi-sectorial. By taking a decomposition of the form \eqref{eq:phase:quasi} into orthogonal sectorial and zero parts, we can see that the inequality is satisfied for some $\epsilon>0$ due to the phase restriction implied by \eqref{eq:result:phase_GD} on the sectorial part. Further, because the phase condition is satisfied on a closed set, we may choose a single $\epsilon>0$ that satisfies the inequality for all $\omega\in[0,\infty]$.	Therefore, the requirements of Theorem \ref{theorem:IQC} are satisfied, and $(G,\Delta)$ is a stable interconnection.
\end{proof}

We recall the robust stability condition provided by the structured singular value $\mu$.
\begin{theorem}\cite{zhou1996}
	Any interconnection $(G,\Delta)$ with $\Delta\in\mathbf{\Delta}_\chi$ for a compatible block dimension $\chi$ satisfying:
	\begin{align} 
		\mu_\chi(G(j\omega)) \cdot \| \Delta(j\omega) \| < 1\, \: \forall \omega\in[0,\infty], \label{eq:problem:mu_sgt}
	\end{align}
	is well-posed and stable.
\end{theorem} 

Our second  main result  combines Theorems 1  and 2 to provide a   mixed  gain  and phase stability criterion, providing a solution to Problem 2.
\begin{theorem} \label{thm:spr:mixed_stability}
	Consider a well-posed feedback system $(G, \Delta)$, where $G, \Delta \in\RHi^{n\times n}$.   Let $\chi=(\mathbf{n},\mathbf{m})$ be a pair of block dimension  tuples compatible with $G$, and let $\mathbf{\Delta}_\chi$ be   a set of  structured perturbations,  with  $\Delta\in\mathbf{\Delta}_\chi$. If  there exist closed sets $\Omega_\psi,\Omega_\mu\subseteq[0,\infty]$ so that $\Omega_\psi\cup\Omega_\mu=[0,\infty]$ and
	\begin{align} 
		\phi(\Delta(j\omega)) + \overline{\psi}_\chi(G(j\omega))&<\pi \: \forall \omega\in\Omega_\psi, \label{eq:spr:system}\\
		\| \Delta(j\omega) \| \cdot \overline{\mu}_\chi(G(j\omega)) &< 1\: \forall \omega\in\Omega_\mu, \label{eq:mu:system}
	\end{align}
	 then  the interconnection $(G,\Delta)$ is stable.
\end{theorem}

\begin{proof}
	We define a piecewise multiplier function $\Pi(j\omega): j\R\rightarrow \C^{2n\times 2n}$. For frequencies $\omega\in\Omega_\psi$, we define $D_\psi(\omega)$ and $\beta(\omega)$ in the same way as in Theorem \ref{thm:spr:phase_stability}, and:
	\begin{align}
		\Pi(j\omega) := \bmat 0 & -e^{-j\beta(\omega)}D_\psi(\omega)^{-1} \\ - e^{j\beta(\omega)}D_\psi(\omega)^{-*} & 0 \emat.
	\end{align}
	Clearly, the FDI on $\Delta$ and $G$ is satisfied at these frequencies. For frequencies $\omega\notin\Omega_\psi$, \eqref{eq:mu:system} is satisfied. At these frequencies, we define $D_\mu(\omega)$ to be a $D$-scaling matrix satisfying:
	\begin{align} \label{eq:result:mu_ineq}
		\| D_\mu(\omega) G(j\omega) D_\mu(\omega)^{-1} \| \leq \Delta(j\omega)^{-1} < \overline{\mu}(G(j\omega)).
	\end{align}
	At these frequencies, we then define:
	\begin{align}
		\Pi(j\omega) := \bmat \overline{\mu}(G(j\omega))^{-2}D(\omega)^*D(\omega) &0\\
		0& -D(\omega)^*D(\omega) \emat
	\end{align}
	From Theorem \ref{thm:spr:phase_stability}, it is clear that $\Delta$ and $G$ satisfy the required FDIs for $\omega\in\Omega_\psi$. At other frequencies, the FDI for $\Delta$ may be written:
	\begin{align}
		\bmat I \\ -\Delta(j\omega)\emat^* \Pi(j\omega) \bmat I \\ -\Delta(j\omega)\emat &\geq 0 \\
		I_n - \overline{\mu}(G(j\omega))^{2} \Delta(j\omega)^*\Delta(j\omega) &\geq 0
	\end{align}
	which is guaranteed by the condition on the singular values of $\Delta(j\omega)$.	Similarly, the FDI on $G$ at these frequencies is:
	\begin{align}
			\bmat G(j\omega) \\ I \emat ^*  \Pi(j\omega) \bmat G(j\omega) \\ I \emat &\leq -\epsilon G(j\omega)^*G(j\omega) \\
			I_n - \tfrac{D(\omega)^{-*}G(j\omega)^* D(\omega)^* D(\omega) G(j\omega) D(\omega)^{-1}}{\overline{\mu}(G(j\omega))^{2}}  &\geq \epsilon G(j\omega)^*G(j\omega)
	\end{align}
	which is guaranteed for all $\omega\in\Omega_\mu$ for some $\epsilon>0$ because \eqref{eq:result:mu_ineq} is satisfied over a closed set. Therefore, the FDI on both $G$ and $\Delta$ is satisfied at all frequencies, and the stability of $(G,\Delta)$ is proved.
\end{proof}

\section{Example} \label{sec:example}
\begin{figure}
	\centering
	\includegraphics{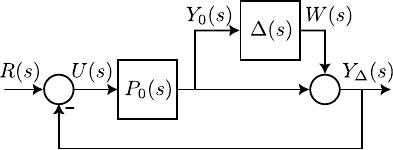}
	\caption{Block diagram of perturbed rotating body system.}
	\label{fig:ex:block}
\end{figure}

We present the application of the combined structured singular value ($\mu$) and structured phase radius ($\psi$) criteria to a simple system. We will additionally analyse the system with an approach combining the structured singular value with the relative passivity index ($R$, detailed in the Appendix), and show that the combination of $\mu$ and $\psi$ can identify a set of admissible perturbations that cannot be identified by a combination of $\mu$ and $R$. We examine the example of a rotating 3 dimensional body given in Chapter 8.6 of \cite{zhou1999}. The dynamics of the system are written in terms of the variables $y_1(t),y_2(t)\in\R$ with inputs $u_1(t),u_2(t)\in\R$. The dynamic equations of the nominal system are given by:
\begin{align}
	\bmat \dot{y}_1 \\ \dot{y}_2 \emat = \bmat 1 & a \\ -a & 1 \emat \bmat y_1 \\ y_2 \emat + \bmat u_1 \\ u_2 \emat,
\end{align}
where $a=(1-I_3/I_1)\Omega$. $I_3$ and $I_1$ are the moments of inertia in the $x$ and $z$ axes. The system may be represented in transfer function form as follows:
\begin{align}
	Y_0(s) = \underbrace{\frac{1}{s^2+a^2}\bmat s - a^2 & a(s+1) \\ -a(s+1) & s-a^2 \emat}_{P_0(s)} U(s)
\end{align}
where $Y_0(s)$ and $U(s)$ are the Laplace transforms of $\bmat y_1 & y_2\emat^\T$ and $\bmat u_1 & u_2 \emat^\T$ respectively, and $P_0(s)$ represents the nominal plant dynamics. The system is perturbed by a multiplicative perturbation $\Delta(s)\in\RHi^{2\times 2}$ as follows:
\begin{align}
	Y_\Delta(s) = (I+\Delta(s)) P(s)U(s) = Y_0(s) + \Delta(s)P(s)U(s).
\end{align}
The system, equipped with the (stabilizing)  control law $U(s)=-Y_\Delta(s)+R(s)$, is illustrated in Figure \ref{fig:ex:block}. The complementary sensitivity transfer function $T(s)$ defined by $Y_0(s)=T(s)W(s)$ is given by:
\begin{align}
	T(s) = \frac{1}{s+1} \bmat 1 & a \\ -a & 1 \emat.
\end{align}
The stability of the perturbed closed loop can be analysed by a feedback loop containing $T$ and $\Delta$ (i.e. we identify $G$ of Figure \ref{fig:iqc_feedback} with $T$.) We consider diagonal perturbations of the form:
\begin{align} \label{eq:Delta}
	\Delta(s) = \bmat 0.5 & 0 \\ 0 & \frac{0.25}{s/b + 1} \emat.
\end{align}
where $b>0$ is a parameter. This perturbation belongs to $\mathbf{\Delta}_\chi$ with $\chi=((), (1,1))$.

\begin{figure}
	\centering
	\includegraphics[width=0.4\textwidth]{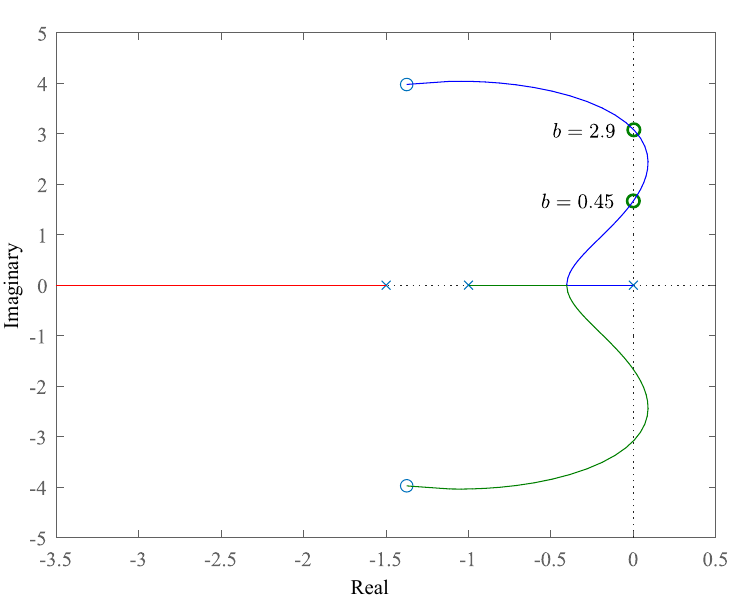}
	\caption{Root locus of perturbed system as $b$ varies from $0$ to $\infty$.}
	\label{fig:ex:rlocus}
\end{figure}

Because the perturbation $\Delta$ is parameterised by a single parameter $b$, we are able to  plot the root locus of the system in Figure \ref{fig:ex:rlocus}. The root locus plot shows that the system is unstable for values of $b\in[0.45,2.9]$. The class of perturbations \eqref{eq:metric:Delta} that we address cannot generally be analysed with a root locus plot--this example allows us to compare the performance of different stability analysis methods by comparing the size of the range of perturbations proven admissible by each method.

\begin{figure}
	\centering
	\includegraphics[width=0.4\textwidth]{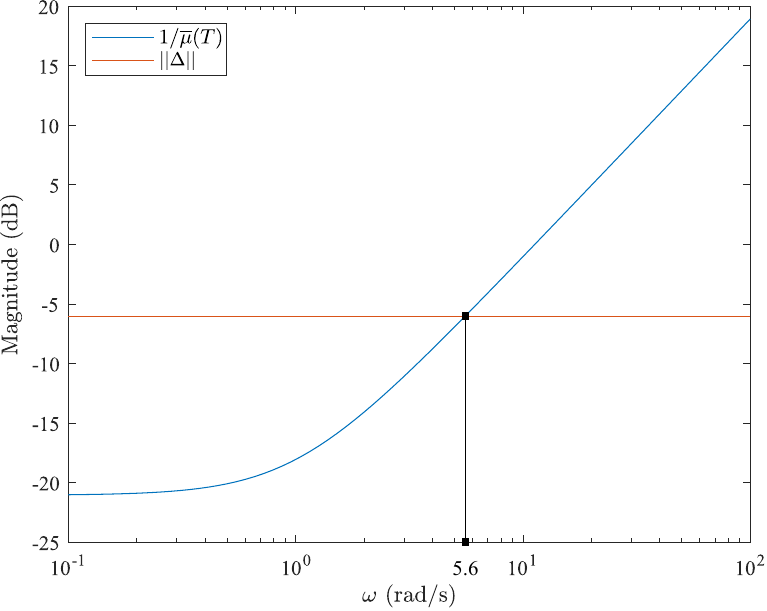}
	\caption{Structured singular value plot for $T$ and singular value plots for $\Delta$.}
	\label{fig:mu_analysis}
\end{figure}
Figure \ref{fig:mu_analysis} plots the reciprocal of $\overline\mu_\chi$ of $T(j\omega)$ (using MATLAB's \texttt{mussv}) and the largest singular value of $\Delta(j\omega)$ in accordance with \eqref{eq:mu:system}. Inspection of $\Delta(s)$ shows that $\| \Delta(j\omega) \|=0.5$ for all $\omega\geq 0$, and therefore independent of $b$. The singular value plots given in Figure \ref{fig:mu_analysis} show that the gain criterion \eqref{eq:mu:system} is satisifed for all frequencies above $5.6$ rad/s.

\begin{figure}
	\centering
	\includegraphics[width=0.4\textwidth]{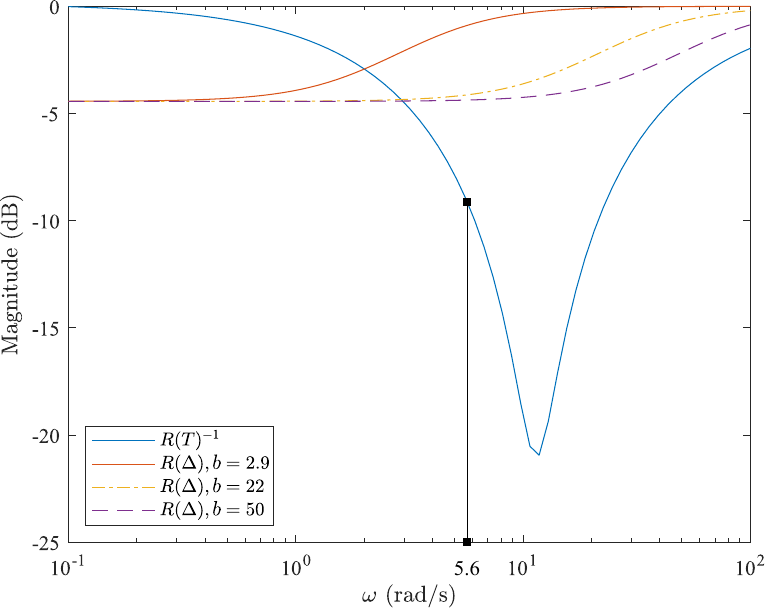}
	\caption{Structured relative passivity index plot for $T$ and relative passivity plot for $\Delta$.}
	\label{fig:R_analysis}
\end{figure}
The reciprocal of the $R$-value of $T(j\omega)$ and the $R$-value of $\Delta(j\omega)$ are plotted in Figure \ref{fig:R_analysis} for several values of $b$ in accordance with \eqref{eq:appendix:R_criterion}. By inspection, we can see that the criterion is satisfied for $\omega$ no greater than 2.9 rad/s regardless of the value of $b$ chosen, due to the asymptotic behaviour of the R-value of $T(j\omega)$ in the low frequency limit. Therefore, regardless of the value of $b$,   neither the structured singular value criterion nor the R-value criterion are satisfied for  $\omega\in(2.9,5.6)$, and stability cannot be concluded for any perturbation of the form \eqref{eq:Delta} on the basis of the combination of these two criteria.

\begin{figure}
	\centering
	\includegraphics[width=0.4\textwidth]{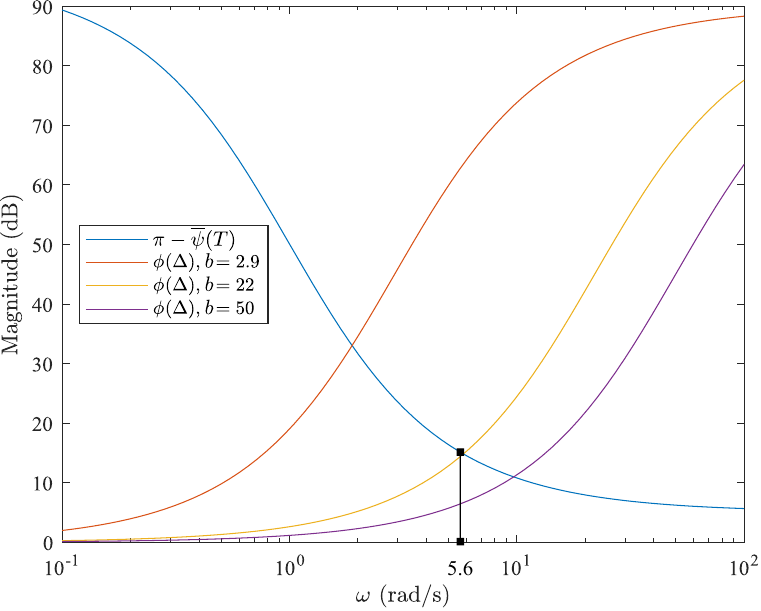}
	\caption{Structured phase radius plot for $T$ and phase radius plots for $\Delta$.}
	\label{fig:phi_analysis}
\end{figure}
The upper bound of the structured phase radius $\overline\psi_\chi$ of $T(j\omega)$ was calculated using the LMI methods detailed in Algorithm \ref{alg1} and plotted alongside the phase radii of $\Delta(j\omega)$. The lower bound $\underline\psi_\chi$ was also calculated as described in Section \ref{sec:spi}, but was found to be equal to $\overline\psi_\chi$ to within 0.1\%, and as such was not plotted. Figure \ref{fig:phi_analysis} shows that for $b\geq 22$, the phase criterion \eqref{eq:spr:system} is satisfied for $\omega<5.6$ rad/s. Therefore, for $b\geq 22$, either the phase or gain criterion is satisfied for all frequencies, guaranteeing the stability of the system.

The root locus plot indicates that the system is stable for perturbations with values of $b>2.9$. The combination of structured singular value analysis and relative passivity analysis could not identify any part of this region of stability, while combining structured singular value analysis with structured phase analysis correctly identifies the portion of the stable region where $b>22$. This demonstrates that there exist  admissible  perturbations that can be identified by  a  combination of  gain  and  phase analysis methods  that could not  have been  identified by combining  gain  and passivity methods.

\section{Conclusion}
We  investigated  the potential  for some  newly developed  notions  of  matrix  phase to provide novel  phase-based  robust stability criteria for stable LTI plants that are subject to  structured perturbations. The approach resembles the well-known $\mu$-analysis  methods that
use matrix  singular  values. When  combined with conventional singular value stability criteria,  the mixed gain and phase-based  methods can  provide  enhanced criteria for determining robust stability  margins. An  example demonstrated that the  novel    criteria were able to determine robust stability criteria for a  class of  perturbations whose stability  could  not  be determined via either a  conventional  $\mu$-analysis,   nor passivity-based methods. Future work could consider the associated synthesis problem, for example optimising structured phase robustness at low frequencies and gain robustness at high frequencies.

\section{Appendix}
\subsection{Derivative of an Eigenvalue}
\begin{theorem}[Theorem 2, \cite{Magnus1985}] \label{thm:maths:d_eig}
	Suppose that $A_0\in\C^{n\times n}$ has a simple eigenvalue $\lambda_0\in\C$ with an associated eigenvalue $u_0\in\C^n$, so that:
	\begin{align}
		A_0u_0 = \lambda_0 u_0.
	\end{align}
	Then there exist functions $\lambda$ and $u$ defined and continuous on a neighbourhood of $A_0$ satisfying:
	\begin{align}
		\lambda(A_0) &= \lambda_0,\\
		u(A_0) &= u_0,\\
		u_0^* u(A) &= 1,\\
		Au(A) &= \lambda(A) u(A).
	\end{align}
	Furthermore, the functions $\lambda$ and $u$ are differentiable on this neighbourhood, and their differentials are:
	\begin{align}
		d\lambda &= \frac{v_0^* (dA) u_0}{v_0^* u_0},\\
		du &= (\lambda_0 I_n - A_0)^+ \left(I- \frac{u_0v_0^*}{v_0^*u_0}\right)(dZ) u_0,
	\end{align}
	where $v_0$ is a right eigenvector of $A_0$ corresponding to $\lambda_0$, i.e:
	\begin{align}
		v_0^*A_0 = \lambda_0 v_0^*,
	\end{align}
\end{theorem}

\subsection{Integral Quadratic Constraints}

\begin{definition}
	Given a bounded self-adjoint operator $\Pi\in\RLi^{2n\times 2n}$, we say that $\Delta:\L_{2e}^n\rightarrow\L_{2e}^n$ satisfies the \emph{IQC (integral quadratic constraint)} defined by $\Pi$ if:
	\begin{align} 
		\int_{-\infty}^\infty \bmat \hat{v}(j\omega) \\ \widehat{\Delta(v)}(j\omega) \emat^*  	\Pi(j\omega) \bmat \hat{v}(j\omega) \\ \widehat{\Delta(v)}(j\omega)\emat d\omega \geq 0  
		 \label{eq:IQC}
	\end{align}
\end{definition}
\vspace{.2cm}

\begin{theorem}{(IQC Stability \cite{Veenman2016})}\label{theorem:IQC}
	Consider the interconnection of ($G$, $\Delta$) given in Figure \ref{fig:iqc_feedback}. Assume:
	\begin{enumerate}
		\item[(I)] For every $\tau\in[0,1]$, the interconnection of $G$ and $-\tau\Delta$ is well-posed.
		\item[(II)] For every $\tau\in[0,1]$, the IQC defined by $\Pi$ \eqref{eq:IQC} is satisfied by $-\tau\Delta$.
		\item[(III)] There exists $\epsilon>0$ such that:
		\begin{align}
			\bmat G(j\omega) \\ I \emat^* \Pi(j\omega) \bmat G(j\omega) \\ I \emat \leq -\epsilon I,\: \forall \omega\in\R.
		\end{align}
	\end{enumerate}
	Then the feedback interconnection of $(G,\Delta)$ is stable.
\end{theorem}
The negative sign applied to $\Delta$ accounts for the negative feedback connection in Figure \ref{fig:iqc_feedback}.

\begin{remark}
	Suppose that the multiplier $\Pi$ is partitioned into $n\times n$ blocks:
	\begin{align}
		\Pi = \bmat \Pi_{11} & \Pi_{12} \\ \Pi_{12}^* & \Pi_{22} \emat.
	\end{align}
	Then, if $\Pi_{11}\geq 0$ and $\Pi_{22}\leq 0$, the IQC defined by $\Pi$ is satisfied by $\tau\Delta$ for $\tau\in[0,1]$ if and only if it is satisfied by $\Delta$ \cite{Veenman2016}. This simplifies condition (II) of Theorem \ref{theorem:IQC}.
\end{remark}

\begin{remark}
	The frequency domain inequality (III) may be replaced with the following \cite{Megretski1997}:
	\begin{align}
		\bmat G(j\omega) \\ I \emat^* \Pi(j\omega) \bmat G(j\omega) \\ I \emat \leq -\epsilon G(j\omega)^*G(j\omega),\: \forall \omega\in\R.
	\end{align}
\end{remark}

\subsection{Relative Passivity Analysis}
Consider the feedback loop interconnection of Figure \ref{fig:iqc_feedback}. A well-known loop transformation allows the plant and perturbation to be replaced with their ``scattering matrix'' equivalents, defined as follows:
\begin{align}
	S_G &= (I-G)(I+G)^{-1},\\
	S_\Delta &= (I-\Delta)(I+\Delta)^{-1}.
\end{align}
The $\mathcal{L}_\infty$ gain of such a ``scattering'' system is no greater than $1$ if and only if $S_G$ is passive \cite{Desoer1975}. For linear systems and perturbations, this permits a frequency domain stability criterion based on the largest singular value of the frequency response of the plant and perturbation. If $\Delta$ comes from the structured set $\mathbf{\Delta}_\chi$, we may use the structured singular value of $S_G$, which we will denote:
\begin{align}
	R(G(j\omega)) = \mu_\chi(S_G(j\omega)).
\end{align}
This leads to a combined structured gain/passivity stability criterion, namely, that if the frequency response of $G$ and $\Delta$ satisfy either one of:
\begin{align} \label{eq:appendix:R_criterion}
	\mu_\chi(G(j\omega)) \cdot \| \Delta(j\omega) \| < 1 \text{ or,}\\
	R(G(j\omega)) \cdot \| S_\Delta(j\omega)\| < 1,
\end{align}
for all $\omega\geq 0$, the interconnection $(G,\Delta)$ is stable.
\printbibliography
%
%
%
%
%
%
%
%

\end{document}